\documentclass[10pt,conference]{IEEEtran}
\ifCLASSINFOpdf
\else
\fi
\usepackage{amsmath,amssymb,amsfonts}
\usepackage{cite}
\usepackage{xcolor}
\usepackage{algorithmic}
\usepackage{mathtools}
\usepackage{bm}
\usepackage{amsthm}
\usepackage[longend,ruled,linesnumbered]{algorithm2e}

\newtheorem{theorem}{Theorem}

\newtheorem{cor}[theorem]{Corollary}

\usepackage{graphicx}
\usepackage{epstopdf}

\IEEEoverridecommandlockouts
\begin{document}

\title{Age-of-Information and Energy Optimization in Digital Twin Edge Networks 
\thanks{This work was supported in part by the grant from the Research Grants Council
(RGC) of the Hong Kong Special Administrative Region, China under Reference No. UGC/FDS16/E02/22, in part by the Research Matching Grant (RMG) under Reference No. CP/2022/2.1, and in part by the Team-based Research Fund under Reference No. TBRF/2024/1.10. This work was completed while Yongna Guo was affiliated with Hong Kong Metropolitan University. (\em{Corresponding author: Yaru Fu})}}
\author{Yongna Guo\IEEEauthorrefmark{1}, Yaru Fu\IEEEauthorrefmark{2}, Yan Zhang\IEEEauthorrefmark{3}, and Tony Q. S. Quek\IEEEauthorrefmark{4} \\
\IEEEauthorrefmark{1}School of Electrical Engineering and Computer Science, KTH Royal Institute of Technology, Stockholm, Sweden\\
\IEEEauthorrefmark{2}School of Science and Technology, Hong Kong Metropolitan University, Hong Kong SAR, China\\
\IEEEauthorrefmark{3}Department of Informatics, University of Oslo, Oslo, Norway\\
\IEEEauthorrefmark{4}Information Systems Technology and Design, Singapore University of Technology and Design, Singapore}

\markboth{}%
{Shell \MakeLowercase{\textit{et al.}}: Bare Demo of IEEEtran.cls for IEEE Journals}

\maketitle

\thispagestyle{empty} 
\begin{abstract}
In this paper, we study the intricate realm of digital twin synchronization and deployment in multi-access edge computing (MEC) networks, with the aim of optimizing and balancing the two performance metrics Age of Information (AoI) and energy efficiency. We jointly consider the problems of edge association, power allocation, and digital twin deployment. However, the inherent randomness of the problem presents a significant challenge in identifying an optimal solution. To address this, we first analyze the feasibility conditions of the optimization problem. We then examine a specific scenario involving a static channel and propose a cyclic scheduling scheme. This enables us to derive the sum AoI in closed form. As a result, the joint optimization problem of edge association and power control is solved optimally by finding a minimum weight perfect matching.
Moreover, we examine the one-shot optimization problem in the contexts of both frequent digital twin migrations and fixed digital twin deployments, and propose an efficient online algorithm to address the general optimization problem. This algorithm effectively reduces system costs by balancing frequent migrations and fixed deployments. Numerical results demonstrate the effectiveness of our proposed scheme in terms of low cost and high efficiency.

\end{abstract}


\IEEEpeerreviewmaketitle

\section{Introduction}

Digital twin has emerged as a promising technology for enabling intelligent Internet-of-Things (IoT) applications, including innovative transportation, industrial manufacturing, smart cities, and more~\cite{Yang2019IN_6G, Letaief2022JSAC, yfu2023, Wu2021IOTJ_Digital}. By facilitating real-time monitoring and interaction with the physical world, digital twin provides a highly accurate digital representation of its physical counterpart~\cite{Lu2021IOTJ, Alcaraz2022COMST_Digital}. Moreover, integrating advanced machine learning tools in digital twin can enhance decision-making and optimization in the digital realm~\cite{ Wang2022IoT_Mobility, Li2023TMC}. 
To realize digital twin in wireless networks, the integration of digital twin within multi-access edge computing (MEC) networks has garnered significant attention. This digital twin edge network is specifically crafted for smooth incorporation into established MEC standardized architectures, paving the way for intelligent 6G edge network services. In particular, this approach offers the advantage of effectively reducing communication delays, ensuring the high-fidelity of digital twin~\cite{Duong2023MWC,Tang2022OJCOMS, Zhang2023JSAC}. 
Besides, digital twin-assisted MEC plays a crucial role in enhancing decision-making capabilities for computation and communication resource allocation in MEC networks~\cite{ Dai2021TII, Van2022TCOM, Van2023JSAC}. 
For example, in \cite{Dai2021TII}, digital twin technology was applied in MEC networks to enhance energy efficiency. This work utilizes machine learning techniques to analyze the digital twin networks and achieve efficient real-time resource allocation. Furthermore, works \cite{Van2022TCOM} and \cite{Van2023JSAC}  leveraged digital twin technology to minimize transmission delays in ultra-reliable and low-latency communications (URLLC) MEC networks.
To enable digital twin-assisted MEC networks, meticulous digital twin synchronization and deployment design is essential to ensure accurate reflections of the physical networks. 
In \cite{Zheng2023TWC}, the user association problem in vehicle networks was investigated to achieve low-latency digital twin synchronization. In addition, considering user mobility, \cite{Lu2021IOTJ} explored digital twin migration to enhance reliability and minimize latency. This was accomplished through the utilization of machine learning techniques to jointly optimize digital twin placement and migration.
Moreover, in data-centric scenarios such as digital twin, age of information (AoI) has become a vital metric for assessing data freshness. Specifically, AoI measures the time elapsed since the most recent data update at a destination, making it crucial for optimizing real-time systems~\cite{abd2019role}. In light of this, the work~\cite{li2023aoi} introduced an optimization framework that prioritizes AoI-aware metrics to enhance service satisfaction. This framework explicitly addressed the digital twin placement problem in digital twin-assisted MEC networks. Furthermore, the work~\cite{zhang2024aoi} proposed the concept of expected AoI to evaluate model accuracy, thereby optimizing the service provisioning of digital twin networks.

In our research, we extend the use of AoI to evaluate the data freshness of digital twins for the purpose of facilitating digital twin synchronization and deployment. In contrast to prior studies, we take into account the energy cost and aim to jointly minimize the weighted sum of AoI and energy cost. The joint optimization of AoI and energy cost provides a comprehensive model for practical systems, addressing both delay-sensitive and energy-sensitive scenarios. For instance, autonomous driving and augmented reality demand stringent delay requirements, while wearable health monitors and smart IoT systems are constrained by energy limitations. Additionally, we consider a unique digital twin synchronization scenario characterized by a substantial volume of IoT devices requiring synchronization. In this context, optimizing edge association, power control, and digital twin deployment is crucial to achieving energy-efficient and high-fidelity digital twin edge networks. 
However, due to the randomness of the channel gains and the interdependence of the optimization variables, the formulated optimization problem is challenging to solve. To tackle it, we first analyze a specific scenario involving static channels where digital twin migration is not required. In this case, we propose a distinct scheduling policy and derive a closed-form sum AoI. By doing so, we can optimally resolve the joint optimization problem of edge association and power control by identifying a minimum weight perfect matching.
Following that, we analyze the optimization problem for a single time slot, considering both frequent digital twin migration and fixed digital twin deployment. Drawing from this analysis, we propose an efficient online lightweight algorithm for the general problem to strike a balance between these two scenarios. Numerical results validate the effectiveness of our proposed online algorithm in terms of average cost.

The rest of this paper is outlined as follows. We present the system model and problem formulation in Section II. The joint optimization problem is investigated in Section III. The numerical results are provided in Section IV, followed by the conclusion and discussion on future work in Section V.

\section{Network Model, AoI, and Problem Formulation}

\subsection{Network Model}

We consider a digital twin-assisted MEC network as illustrated in Fig.~\ref{system}, which includes $M$ edge servers, indexed by $\mathcal{M}\triangleq\{1,2,\ldots, M\}$, and $K$ IoT devices, indexed by $\mathcal{K}
\triangleq\{1,2,\ldots,K\}$, where $K>>M$. In the beginning, the digital twin for each device is randomly initiated and deployed in the edge servers. Each IoT device has its corresponding digital twin deployed on a specific edge server. To maintain the fidelity of the digital twin service, the devices must synchronize the latest data with their digital twin at the edge servers. Here, the synchronized data can vary across different physical devices, encompassing environmental data, maintenance data, and more. We consider a discrete-time system in which the whole time horizon is divided into equal time slots denoted as $\mathcal{T}\triangleq\{1,2,\ldots,T\}$. Each time slot has a fixed duration of $T_s$. 
Let $b_{k,m}^t$ represent the digital twin deployment indicator, which is a binary variable, and $b_{k,m}^t=1$ if the digital twin of device~$k$ is deployed in server~$m$ at time slot~$t$, and $b_{k,m}^t=0$ otherwise. Besides, in each time slot, each device has a dedicated digital twin that is deployed exclusively on one specific server, i.e., $\sum_{m=1}^M b_{k,m}^t = 1, k\in\mathcal{K}, t\in\mathcal{T}.$    
In each time slot, the devices are scheduled to synchronize their data to the edge servers where their digital twins are deployed.
Let $h_{k,m}^t$ denote the channel power gain between IoT device~$k$ and edge server~$m$ in time slot~$t$, which includes both large-scale fading and small-scale fading.

\subsubsection{Edge association for synchronization}

We assume that for each edge server, at most one device is scheduled to transmit to this specific edge server at one time slot, and at most $M$ devices are scheduled to transmit to $M$ edge servers simultaneously in one time slot. When a device is scheduled to synchronize in one time slot, it will sample the latest data and transmit it to its associated edge server.
Let $a_{k,m}^t$ be a binary variable for edge association and $a_{k,m}^t=1$ if device~$k$ is scheduled to transmit data to edge server~$m$ in time slot~$t$, and $a_{k,m}^t=0$ otherwise, where $t\in\mathcal{T}$.
The transmit data size of device~$k$ is denoted by $D_k$. 
If device~$k$ is scheduled to transmit to server~$m$ in time slot~$t$, let the transmit rate be denoted by $R_{k,m}^t$, and
\begin{equation}\label{rate}
    R_{k,m}^{t} = B \log (1+\frac{{h_{k,m}^{t}} p_{k,m}^{t}}{\sigma^2}),
\end{equation}
where $B$ is the channel bandwidth, $p_{k,m}^{t}$ is the transmit power of device~$k$ to server~$m$ in time slot~$t$, and $\sigma^2$ is the noise power. Moreover, when a device is scheduled to transmit in time slot~$t$, it is required to complete its transmission within the time slot, i.e., 
\begin{equation}\label{time}
    \frac{D_k}{R_{k,m}^{t}}\leq T_s, \;
\end{equation}
where the transmit power can be adjusted to satisfy the time constraint.
Define $E_k^{t}$ as the transmit energy of device~$k$ in time slot~$t$, which is quoted below:
\begin{equation}\label{e_k_t}
     E_k^{t} = 
    \sum_{m=1}^M a_{k,m}^t \frac{D_k p_{k,m}^{t}}{ R_{k,m}^{t}}.
\end{equation}

\subsubsection{Backhaul transmission}
When the associated server~$m'$ of device~$k$ is different from its digital twin deployment server~$m$, an additional transmission is needed. In this case, apart from the transmission between device~$k$ and server~$m'$, server $m'$ is also responsible for transmitting the data of device $k$ to server $m$. This transmission between the two edge servers incurs an additional energy cost of $\eta$ per bit. The extra delay cost for the backhaul communication is negligible compared to the scheduling delay. With the above-mentioned definitions, the backhaul energy cost for time slot~$t$ is 
\begin{equation}
    E_{\text{back}}^{t} = \sum_{m=1}^M\sum_{k=1}^K a_{k,m}^{t}(a_{k,m}^{t}-b_{k,m}^{t})\eta D_k.
\end{equation}
\begin{figure}
    \centering
    \includegraphics[width=0.41\textwidth]{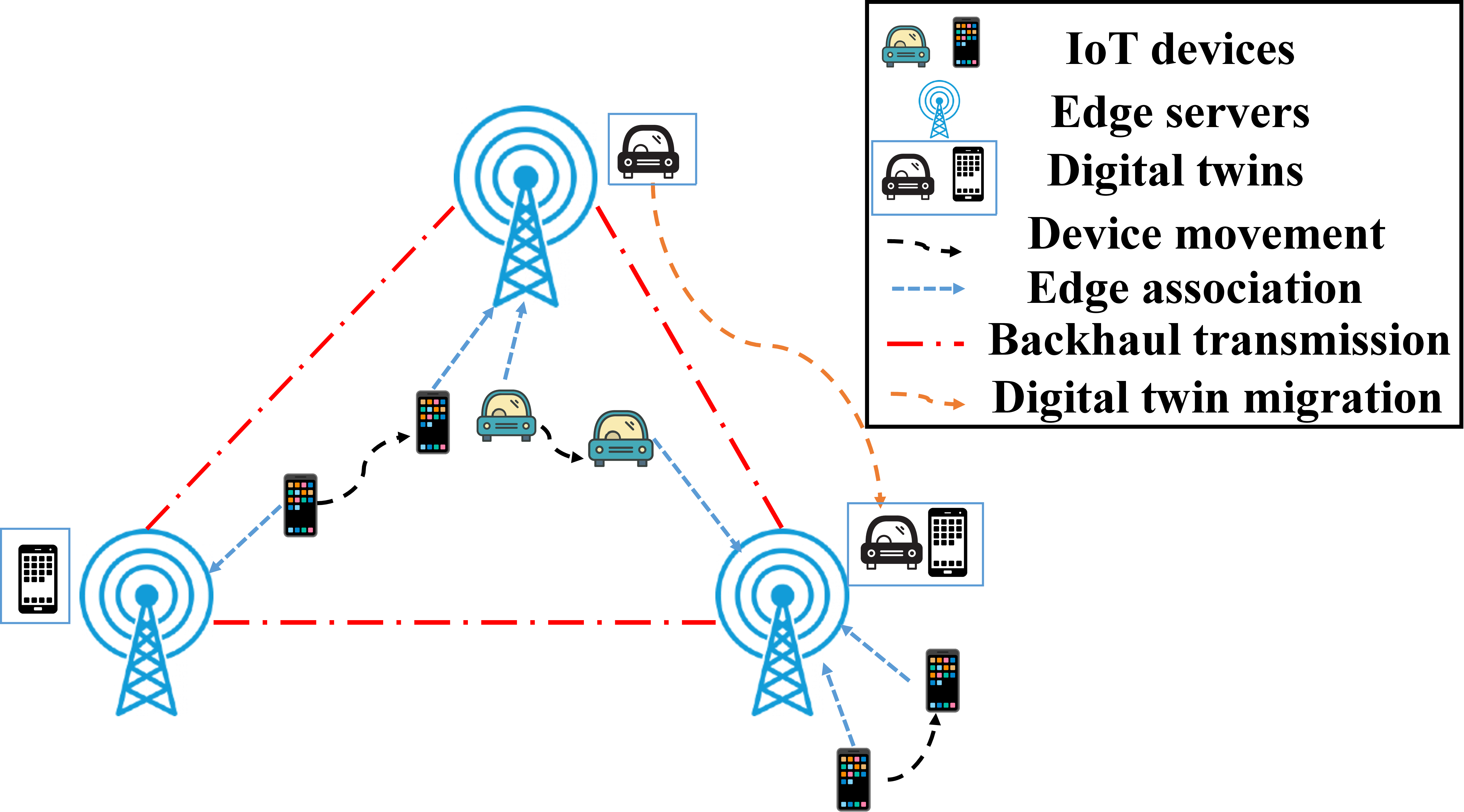}
    \caption{An illustration of digital twin edge networks.}
    \label{system}
\end{figure}
\subsubsection{Digital twin migration}
To eliminate backhaul costs, we consider the digital twins of devices may migrate among different servers in different time slots to serve the synchronization of the devices. Note that there will be extra energy costs due to the digital twin migration when the deployment of digital twins of the devices in the current time slot differs from the previous one. Define the energy cost per bit for each digital twin migration as $\lambda$. Let digital twin bit size for device~$k$ be denoted by $\tilde{D}_k$. 
The total energy cost for digital twin migration in time slot~$t$ is
\begin{equation}
    E_{\text{mig}}^{t} = \frac{1}{2} \sum_{k=1}^K \sum_{m=1}^M \left|b_{k,m}^{t}-b_{k,m}^{t-1}\right|\lambda \tilde{D}_k.
\end{equation}

\subsection{Age of Information}
In this paper, we utilize AoI to evaluate the fidelity of the digital twin service. The AoI for the digital twin of device~$k$ is defined as the time elapsed since its most recent stored data was generated. We evaluate the AoI in terms of the number of time slots elapsed. Let $\Delta_k(t)$ represent the AoI of digital twin of device~$k$ in time slot~$t$. We initialize $\Delta_k(1)=1$ for $k\in\mathcal{K}$. The AoI with regard to digital twin is determined by both the AoI of the previous time slot and whether synchronization occurred in that slot. 
If there is no synchronization in the previous slot, this condition can be expressed as $\sum_{m=1}^M a_{k,m}^{t-1} = 0$. In this case, the AoI of digital twin increases by one when the slot elapses by one, resulting in the following equation:
\begin{equation}\label{aoi1}
    \Delta_k(t) = \Delta_k(t-1) + 1,\; \text{when }\; \sum_{m=1}^M a_{k,m}^{t-1} = 0.
\end{equation}
If the synchronization occurs in the last slot, it can be represented as $\sum_{m=1}^M a_{k,m}^{t-1} = 1$. 
In that case, the AoI is reset to 1, i.e.,
\begin{equation}\label{aoi2}
 \Delta_k(t) = 1, \; \text{when }\; \sum_{m=1}^M a_{k,m}^{t-1} = 1.
\end{equation}
We combine~\eqref{aoi1} and \eqref{aoi2} as follows:
\begin{equation}\label{aoi}
    \Delta_k(t) =
     (1-\sum_{m=1}^M a_{k,m}^{t-1})\Delta_k(t-1) + 1.
\end{equation}
Further, to keep the fidelity of the digital twin edge networks, the AoI of each digital twin should not exceed the maximum AoI $\Gamma$, i.e.,
\begin{equation}\label{aoi_gamma}
    \Delta_k(t) \leq \Gamma, k\in\mathcal{K}, t\in\mathcal{T}.
\end{equation}

\subsection{Problem Formulation}
We aim to jointly optimize edge association, digital twin deployment, and power allocation to minimize the weighted sum of the long-term average AoI and energy, which are denoted by $\overline{\Delta}$ and $\overline{E}$, respectively. With the definitions, we have
\begin{equation}
    \overline{\Delta} \triangleq \frac{1}{KT}\sum_{t=1}^T\sum_{k=1}^K \Delta_k(t),
\end{equation}
and 
\begin{equation}
    \overline{E}\triangleq\frac{1}{KT}\sum_{t=1}^T(E_{\text{back}}^{t}+E_{\text{mig}}^{t}+ \sum_{k=1}^K E_{k}^{t}).
\end{equation}
Let $\bm{a}^{t} = (a_{k,m}^{t})_{k\in\mathcal{K},m\in\mathcal{M}}$ be the association relationship between all devices and servers in time slot~$t$, and let $\bm{a} = (\bm{a}^{t})_{t\in\mathcal{T}}$ be the association relationship between all devices and servers throughout the entire time horizon. Similarly, let $\bm{b}^{t} = (b_{k,m}^{t})_{k\in\mathcal{K},m\in\mathcal{M}}$ and $\bm{b} = (\bm{b}^{t})_{t\in\mathcal{T}}$ denote the digital twin deployment relationship between all devices and servers in time slot~$t$ and throughout the entire time horizon, respectively. Let $\bm{p}^t = (p_{k,m}^t)_{k\in\mathcal{K}, m\in\mathcal{M}}$ and $\bm{p} = (\bm{p}^t)_{t\in\mathcal{T}}$ be the transmit power of device~$k$ in time slot~$t$ and throughout the entire time horizon, respectively. The optimization problem, denoted as \textbf{P1}, can be mathematically formulated as follows:
\begin{equation*}
   \textbf{P1} \; :    \min_{\bm{a}, \bm{b},\bm{p}}\; \xi \overline{\Delta} + (1-\xi) \overline{E}
\end{equation*}
subject to 
\begin{equation*}\label{const}
    \begin{split}
        C1:& \; a_{k,m}^{t}, b_{k,m}^{t} \in \{0,1\},  k \in \mathcal{K}, m \in \mathcal{M}, t \in \mathcal{T}\\  
        C2:& \; \sum_{k=1}^K a_{k,m}^{t} \leq 1,  m\in\mathcal{M},  t \in \mathcal{T}, \\
       C3 : &\; \sum_{m=1}^M b_{k,m}^t = 1, k\in\mathcal{K}, t\in\mathcal{T}, \\
        C4:& \; \eqref{time} \; \text{and} \;  \eqref{aoi_gamma},
    \end{split}
\end{equation*}
where $0\leq \xi \leq 1$ depicts the system requirements on delay and energy. The consideration of minimizing the weighted sum of AoI and energy aims to provide a comprehensive formulation that addresses diverse system requirements. Specifically, when $\xi=0$, the system is energy-limited and can tolerate a higher AoI. Conversely, when $\xi=1$, it signifies that the system has strict requirements on AoI and prioritizes minimizing AoI over energy consumption.
$C1$ represents the binary constraint on the association relationship variables and the digital twin deployment variables. $C2$ restricts that during each time slot, at most one device can be scheduled to transmit to each server. $C3$ indicates that each device has one digital twin deployed on a specific server. $C4$ represents the delay constraint on the transmission time for each scheduling and the constraint on the maximum AoI. {Before closing this section, it is important to note that in practical systems, the units and scales of AoI and energy may differ. To tackle this challenge, normalizing both AoI and energy by their average values can ensure consistency and comparability.}

\ifCLASSOPTIONcaptionsoff
  \newpage
\fi

\section{Algorithm Design}
Before investigating the developed algorithm, we first analyze the feasibility condition for problem~\textbf{P1}. 
\begin{cor}\label{cor1}
    \textbf{P1} is feasible if and only if $K\leq M\Gamma$.
\end{cor}
\begin{proof}
    Following the maximum AoI constraint \eqref{aoi_gamma}, each device should be scheduled at least once during every $\Gamma$ time slots. To realize this, $K$ should be no greater than $M\Gamma$ since at most $M$ devices are scheduled during each time slot. Otherwise, at least one device will not be scheduled in $\Gamma$ time slots, and its AoI will be larger than $\Gamma$. When $K\leq M\Gamma$, one feasible solution is to schedule $M$ devices sequentially in each time slot, ensuring that all devices are scheduled at least once within $\Gamma$ time slots. This scheduling procedure is repeated in the following $\Gamma$ time slots.
\end{proof}

However, the general optimization problem $\textbf{P1}$ is challenging to solve due to the randomness of the channel gains and the coupling of the optimization variables. To tackle these challenges, we first analyze the optimization problem under the static channels, where DT migration is not required. Afterward, we investigate the optimization problem in one time slot, considering both the frequent digital twin migration and fixed digital twin deployment scenarios. Based on the analysis, we develop an online lightweight algorithm for the general problem to strike a balance between frequent digital twin migration and fixed digital twin placement.

\subsection{An Optimal Algorithm for the Special Case with $K=M\Gamma$ and a Static Channel}
We define the scheduling policy of $\Gamma$ time slots as the association relationship of these $\Gamma$ time slots, i.e., $\bm{\pi} \triangleq (\bm{a}^1, \ldots, \bm{a}^\Gamma)$.
\begin{cor}\label{cor2}
    When $K=M\Gamma$, the devices should be scheduled following one specific scheduling policy $\bm{\pi} $ every $\Gamma$ time slots, while satisfying $\sum_{k=1}^K a_{k,m}^t = M$ and $\sum_{t=1}^\Gamma a_{k,m}^t = 1$.  
\end{cor}
\begin{proof}
    As stated in the proof of Corollary~\ref{cor1}, each device should be scheduled at least once per $\Gamma$ time slots. Since $K=M\Gamma$, to ensure the above requirement, each device should be scheduled exactly once every $\Gamma$ time slots, and in each time slot, $M$ devices should be scheduled, i.e., $\sum_{k=1}^K a_{k,m}^t = M$ and $\sum_{t=1}^\Gamma a_{k,m}^t = 1$. Assume that during the first $\Gamma$ time slots, the devices are scheduled according to policy~$\bm{\pi}$. At time slot~$\Gamma+1$, the devices scheduled at time slot one have AoI values of $\Gamma$, indicating that they should be scheduled at time slot~$\Gamma+1$. Following the same procedure, the devices should be scheduled during the second $\Gamma$ time slots following the same scheduling policy~$\bm{\pi} $. Similarly, the devices should be scheduled following one specific scheduling policy $\bm{\pi} $ every $\Gamma$ time slots.
\end{proof}
Referring to Corollary~\ref{cor2}, to optimize the problem under $K=M\Gamma$, we only need to determine the scheduling policy $\bm{\pi}$. Given static channels where the channel gains of all users remain constant across all time slots, digital twin migration is deemed unnecessary for minimizing overall energy consumption. Besides, digital twin initialization can be optimized to eliminate the cost of backhaul synchronization by letting the digital twin placement server of each device be the same as the associated server, i.e., $b_{k,m}^t = a_{k,m}^t$. Thus, the optimization problem~\textbf{P1} in $T$ time slots can be simplified as the scheduling and power allocation problem in the first $\Gamma$ time slots, assuming $K=M\Gamma$ and static channel. This simplified problem can be mathematically formulated as~\textbf{P2}:
\begin{equation*}
   \textbf{P2} \; :    \min_{\bm{\pi},\bm{p}}\; \sum_{t=1}^{\Gamma} \sum_{k=1}^K\xi \Delta_k^t + (1-\xi) E_k^t, 
\end{equation*}
subject to 
\begin{equation*}
    C1, C2, \eqref{time}, \eqref{aoi_gamma}, \text{and} \; \sum_{m=1}^M a_{k,m}^t =1, k\in{\mathcal{K}}, t\in\{1,\ldots, \Gamma\}.
\end{equation*}

\begin{cor}\label{cor3}
    For any scheduling policy~$\bm{\pi}$, the total AoIs for all devices across the $\Gamma$ time slots, denoted by $\Delta_{\text{sum}}$, remains constant and is given by $\Delta_{\text{sum}}= \frac{M(2\Gamma^3+3\Gamma^2+\Gamma)}{6}$. 
\end{cor}
\begin{proof}
    When device~$k$ is scheduled in time slot~$t'$ during the $\Gamma$ time slots, the sum of AoIs in all $\Gamma$ time slots with regard to device~$k$, denoted as $\Delta_{\text{sum},k} $, is
    \begin{equation}
        \Delta_{\text{sum},k}  = \sum_{t=1}^{\Gamma} \Delta_k^t  =\sum_{i=1}^{t'} i + \sum_{j=1}^{\Gamma-t'} j = t'^2 - \Gamma t' + \frac{\Gamma^2}{2} + \frac{\Gamma}{2}.
    \end{equation}
    By Corollary~\ref{cor2}, each device is scheduled exactly once in $\Gamma$ time slots. Besides, in each time slot, there are $M$ devices scheduled. Therefore, we can derive 
    \begin{equation}
    \begin{split}
        \Delta_{\text{sum}} &
 = \sum_{k=1}^K \Delta_{\text{sum},k}= M\sum_{t'=1}^\Gamma (t'^2 - \Gamma t' + \frac{\Gamma^2}{2} + \frac{\Gamma}{2}) \\& = \frac{1}{6}{M(2\Gamma^3+3\Gamma^2+\Gamma)}.\qedhere 
    \end{split}  
    \end{equation}
    \end{proof}

In accordance with Corollary~\ref{cor3}, problem \textbf{P2} degenerates to the energy minimization problem, which can be represented as a bipartite graph~$G = (U, V, E)$, where $U=\mathcal{K}$ represents the set of $K$ users, $V$ represents the set for $M\Gamma$ edge servers, and $E$ represents the set of the edges connecting each vertex in $U$ and each vertex in $V$. Since we consider the scheduling problem in $\Gamma$ time slots, the set $V$ is generated by duplicating the $M$ edge servers $\Gamma$ times. Besides, let each vertex~$v_{t,m}\in V$ be labelled by $(t,m), t\in\{1,2,\ldots, \Gamma\}, m\in\mathcal{M}$. The edge weight between $u\in{U}$ and $v_{t,m}\in V$ is defined as
\begin{equation}
    w(e_{u,v_{t,m}})=E_{u,m}^t,
\end{equation}
$E_{u,m}^t$ represents the minimum energy for device~$u$ transmitting data to server~$m$ in time slot~$t$, which can be determined by Corollary~\ref{cor_power}. The value of $E_{u,m}^t$ remains constant across different time slots under static channels.
\begin{cor}\label{cor_power}
    If device~$k$ is associated with server~$m$ at time slot~$t$, the optimal power to minimize $E_{k,m}^t$ is 
    \begin{equation}\label{p_opt}
        p_{k,m}^{t} = \frac{\sigma^2(2^{D_k/(BT_s)}-1)}{h_{k,m}^t}.
    \end{equation}
\end{cor}
\begin{proof}
    Based on \eqref{e_k_t}, we have 
    \begin{equation}
        E_{k,m}^t = \frac{D_kp_{k,m}^t}{B\log(1+{h_{k,m}^tp_{k,m}^t}/{\sigma^2})}.
    \end{equation}
    By taking the derivative with respect to $p_{k,m}$, we obtain
    \begin{equation}
    \begin{split}
    &\frac{d}{dp_{k,m}^t}E_{k,m}^t = \\
&\frac{\ln{(2)}D_k\left((h_{k,m}^tp_{k,m}^t+\sigma^2)\ln{(\frac{h_{k,m}^tp_{k,m}^t}{\sigma^2}+1)} - h_{k,m}^tp_{k,m}^t\right)}{B(h_{k,m}^tp_{k,m}^t+\sigma^2)\ln{(\frac{h_{k,m}^tp_{k,m}^t}{\sigma^2}+1})},
        \end{split}
    \end{equation}
which is greater than zero. Thus, $E_{k,m}^t$ grows monotonically with $p_{k,m}^t$. Besides, to satisfy \eqref{time}, we have
    \begin{equation}
        p_{k,m}^t \geq \frac{\sigma^2(2^{D_k/(BT_s)}-1)}{h_{k,m}^t}.
    \end{equation}
    Therefore, the optimal power is given by \eqref{p_opt}.
\end{proof}
\begin{theorem}
   Problem \textbf{P2} can be solved optimally by finding the minimum weight perfect matching~$M$ in graph~$G$, i.e., $w(M)=\sum_{e\in M}w(e)$. 
\end{theorem}

\begin{proof}
    By Corollary~\ref{cor3}, problem~\textbf{P2} degenerates to the energy minimization problem. Since $K=M\Gamma$, $M$ users should be scheduled to $M$ servers in each time slot, and the $K$ users should be scheduled exactly once during the $\Gamma$ time slots, which is equivalent to the perfect matching problem in $G$. Besides, by the definition of edge weights, it is clear that a minimum weight perfect matching minimizes the total energy consumption.
\end{proof}

\subsection{Analysis for the Optimization Problem in One Time Slot}
When approaching the optimization problem for a single time slot, the goal is to minimize the weighted sum of AoI and energy. Since $\Delta_k(t)$ is determined by the status of the previous time slot and remains fixed in time slot~$t$, we only need to consider the energy consumption in this time slot. For brevity, we can remove the constant factor $(1-\xi)$ and the one-shot optimization problem is formulated as follows:
\begin{equation*}
   \textbf{P3} \; :    \min_{\bm{a}^t, \bm{b}^t,\bm{p}^t}\;   E_{\text{back}}^{t}+E_{\text{mig}}^{t}+ \sum_{k=1}^K E_{k}^{t},
\end{equation*}
subject to 
\begin{equation*}
    C1, C2, C3, C4,\; \text{and}\;\eqref{time}.
\end{equation*}
\begin{cor}\label{cor4}
    Let $\Tilde{{\mathcal{K}}}_t \triangleq \{k: \Delta_k(t) = \Gamma\}$. The feasibility condition of \textbf{P3} is $|\Tilde{{\mathcal{K}}_t} | \leq M$.
\end{cor}
\begin{proof}
 For the devices whose AoI is equal to the maximum AoI constraint, i.e., $k\in\Tilde{\mathcal{K}}_t$, it is necessary to schedule them in time slot $t$. Failing to schedule these devices in the given time slot would result in their AoI exceeding the constraint. Moreover, due to $C2$, the maximum number of devices scheduled in each time slot cannot exceed the number of servers $M$.  
\end{proof}
\begin{cor}\label{p3_cor1}
If device~$k$ satisfies $\Delta_k(t)<\Gamma$, the optimal solution for device~$k$ in \textbf{P3} is $a_{k,m}^t=0, \forall m$, $b_{k,m}^t=b_{k,m}^{t-1}, \forall m$, and $p_k=0$.
\end{cor}
\begin{proof}
    It is apparent that to minimize energy consumption, it suffices to transmit only the data of devices that require synchronization, i.e., $k\in\Tilde{\mathcal{K}}_t$. Thus, there is no need to transmit the data of the devices that satisfy the AoI constraint.
\end{proof}

By Corollaries~\ref{cor4} and \ref{p3_cor1}, we can reformulate \textbf{P3} by considering only the user pairing and digital twin placement of the devices in $\Tilde{\mathcal{K}}_t$. 
Nonetheless, optimizing the problem remains complex due to the interdependence of the optimization variables $\bm{a}^t$ and $\bm{b}^t$. Therefore, we consider the following two different scenarios given different assumptions on $\bm{b}^t$:
\subsubsection{Scenario 1}
We assume that in each time slot, the digital twin for a given device is deployed at its associated edge server, i.e., $\bm{a}^t = \bm{b}^t$. As a result, no backhaul transmission cost was incurred. The reformulated problem~\textbf{P3.1} is stated as follows:
\begin{equation*}
   \textbf{P3.1} \; :    \min_{\bm{a}^t,\bm{p}^t}\;  \sum_{k\in\Tilde{\mathcal{K}}_t}  \big( 
   \frac{1}{2} \sum_{m=1}^M \left|a_{k,m}^{t}-a_{k,m}^{t-1}\right|\lambda \tilde{D}_k+  E_{k}^{t}\big),       
\end{equation*}
subject to
\begin{equation*}\label{con_3_1}
    C1, C2, C3, C4, \eqref{time}, \text{and} \;C7: \sum_{m=1}^M a_{k,m}^t =1, k\in\Tilde{\mathcal{K}_t}.
\end{equation*}

\textbf{P3.1} can be represented by a bipartite graph~$G'=\{U', V', E'\}$, where $U'=\tilde{K}_t$, $V'$ represents the set of the $M$ edge servers, and $E'$ represents the set of the edges. The edge weight between $u\in U'$ and $v \in V'$ is denoted by
\begin{equation}
    w'(e_{uv}) =(1-b_{u,v}^{t-1})\lambda \tilde{D}_u+  E_{u,v}^{t}, 
\end{equation}
where $E_{u,v}^{t}$ represents the minimum transmit energy between device~$u$ to server~$v$, determined by Corollary~\ref{cor_power}. 
Besides, we add $M-\tilde{K}_t$ dummy devices to $U'$ with the weight of the edge connecting each dummy device and each server in $V'$ set to a significant value. 
\begin{theorem}\label{3_1}
    Problem~\textbf{P3.1} can be optimally resolved by identifying the minimum weight perfect matching $M'$ in graph~$G'$.
\end{theorem}

\begin{proof}
    By the definition of the edge weight in the bipartite graph representation, it is clear that we can obtain the optimal solution for problem \textbf{P3.1} by finding the minimum weight perfect matching.
\end{proof}

\subsubsection{Scenario 2}

To avoid frequent digital twin migration, we presume that the deployment of digital twins remains consistent across various time slots, i.e., $\bm{b}^t = \bm{b}^{t-1}$. Consequently, there is no migration cost involved. This situation enables the reformulation of the problem as \textbf{P3.2}:
\begin{equation*}
    \textbf{P3.2}\; :\min_{\bm{a}^t, \bm{p}^t} \sum_{k\in\Tilde{\mathcal{K}}_t} \big( \sum_{m=1}^M a_{k,m}^{t}(a_{k,m}^{t}-b_{k,m}^{t})\eta D_k+  E_{k}^{t}\big),       
\end{equation*}
where the constraints are the same as that in \textbf{P3.1}. Similar to \textbf{P3.1}, problem \textbf{P3.2} can also be represented by a bipartite graph~$\tilde{G}$. In this graph, the edge weight is given below: 
\begin{equation}
    \tilde{w}(e_{uv}) = (1-b_{u,v}^{t-1})\eta D_u+  E_{u,v}^{t}.
\end{equation}
Likewise, \textbf{P3.2} can be optimally addressed by identifying the minimum weight perfect matching $\tilde{M}$ within $\tilde{G}$.

\subsection{An Online Lightweight Algorithm for the General Case}

In this subsection, we propose an online lightweight scheduling algorithm for the general problem, where digital twin migration is only executed under specific criteria. As shown in \textbf{P3.1}, digital twin migration is conducted in every time slot. This frequent migration, however, can substantially increase the overall system's energy consumption. Conversely, in \textbf{P3.2}, the digital twin deployment remains unchanged across time slots. While this fixed deployment reduces migration costs, it may lead to significant accumulated backhaul expenses due to device mobility. To mitigate those issues, we introduce a scheduling algorithm to strike a balance between the migration cost and the backhaul cost. The core principle for this scheduling algorithm is: The migration will only occur when the overall accumulated backhaul cost surpasses the weighted migration energy cost, i.e., $E_{\text{back}}^{\text{sum}}\leq \beta E_{\text{mig}}^t$, where the overall accumulated backhaul cost is defined as $E_{\text{back}}^{\text{sum}} \triangleq \sum_{s=t'}^t E_{\text{back}}^s$, $t'$ and $t$ represent the time slot of the last occurrence of digital twin migration and the current time slot, respectively, and $\beta$ is the weight factor. We let $\bm{b}^t = \bm{a}^t$, and the overall migration energy cost~$E_{\text{mig}}^t$ can be obtained from Theorem~\ref{3_1}. Besides, the scheduling problem should be jointly considered across multiple time slots to satisfy the maximum AoI constraint. Therefore, we propose a time-efficient cyclic scheduling algorithm to satisfy the AoI constraint. There is an inherent tradeoff between the AoIs and the energy consumption, i.e., frequent synchronization can reduce the average AoIs but increase the energy cost. To reduce energy consumption, we let the cyclic period be $\Gamma$ time slots, balancing the need to avoid frequent synchronization while maintaining the AoI constraint. The dynamic adjustment of the cyclic period will be explored in our future endeavors. Let the general cyclic scheduling sequence for every $\Gamma$ time slots be denoted by ${\bm{\Pi}}$, which can be initialized following a specific order. Besides, let the set of devices scheduled on time slot~$t$ following ${\bm{\Pi}}$ be denoted by ${\mathcal{K}}_{t}^{{\bm{\Pi}}}$. The proposed lightweight online scheduling algorithm is summarized in Algorithm~1.
The complexity of Algorithm~1 in each iteration is the same as the sum of the complexity of \textbf{P3.1} and \textbf{P3.2}. Both are equivalent to finding the maximum weight perfect matching in a bipartite graph, which can be efficiently solved using the Hungarian algorithm with a complexity of $O(V^3)=O(K^3)$~\cite{papadimitriou2013combinatorial}. Consequently, the overall complexity of Algorithm~\ref{alg1} is $O(TK^3)$, making it efficient for practical applications.

\begin{algorithm}\label{alg1}
\small
\caption{Online lightweight algorithm for \textbf{P1}}
\SetAlgoNlRelativeSize{0} 
\SetNlSty{textbf}{}{} 
\SetKwInOut{KwIn}{Input} 
\SetKwInOut{KwOut}{Output} 
\KwIn{Cyclic scheduling policy ${\bm{\Pi}}$;}
\KwOut{$\bm{a}, \bm{b}, \bm{p}$;}
Initialize $E_{\text{back}}^{\text{sum}} := 0$\;
\For{$t = 1, 2, \ldots, T$}{
    Let $\bm{b}^t = \bm{a}^t$, $\tilde{\mathcal{K}}_t = \mathcal{K}_t^{{\bm{\Pi}}}$, and obtain $\bm{a}^t$, $\bm{p}^t$, and $E_{\text{mig}}^t$ by solving \textbf{P3.1}\;
    \If{$E_{\text{back}}^{\text{sum}} \leq \beta E_{\text{mig}}^t$}{
        Let $\bm{b}^t = \bm{b}^{t-1}$, $\tilde{\mathcal{K}}_t = \mathcal{K}_t^{{\bm{\Pi}}}$, and obtain $\bm{a}^t$, $\bm{p}^t$, and $E_{\text{back}}^t$ by solving \textbf{P3.2}\;
        Let $E_{\text{back}}^{\text{sum}} := E_{\text{back}}^{\text{sum}} + E_{\text{back}}^t$\;
    }
    \Else{
        Let $E_{\text{back}}^{\text{sum}} := 0$\;
    }
}
\end{algorithm}

\section{Numerical Results}

\begin{figure*}[t]
    \centering    
    \begin{minipage}[t]{0.24\textwidth}
        \includegraphics[width=\textwidth]{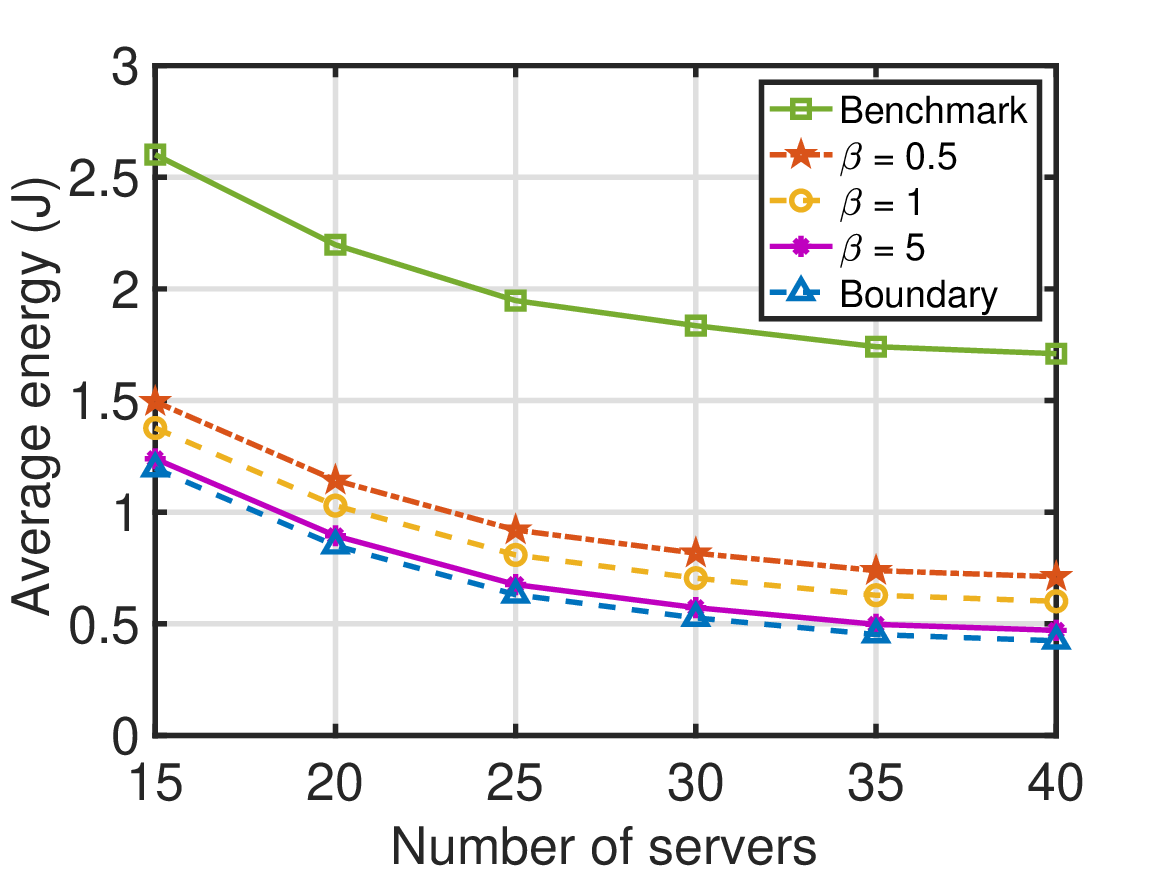}
        \caption{Average energy vs. The number of servers, $M$}
        \label{dt1}
    \end{minipage}
    \hfill
    \begin{minipage}[t]{0.24\textwidth}
        \includegraphics[width=\textwidth]{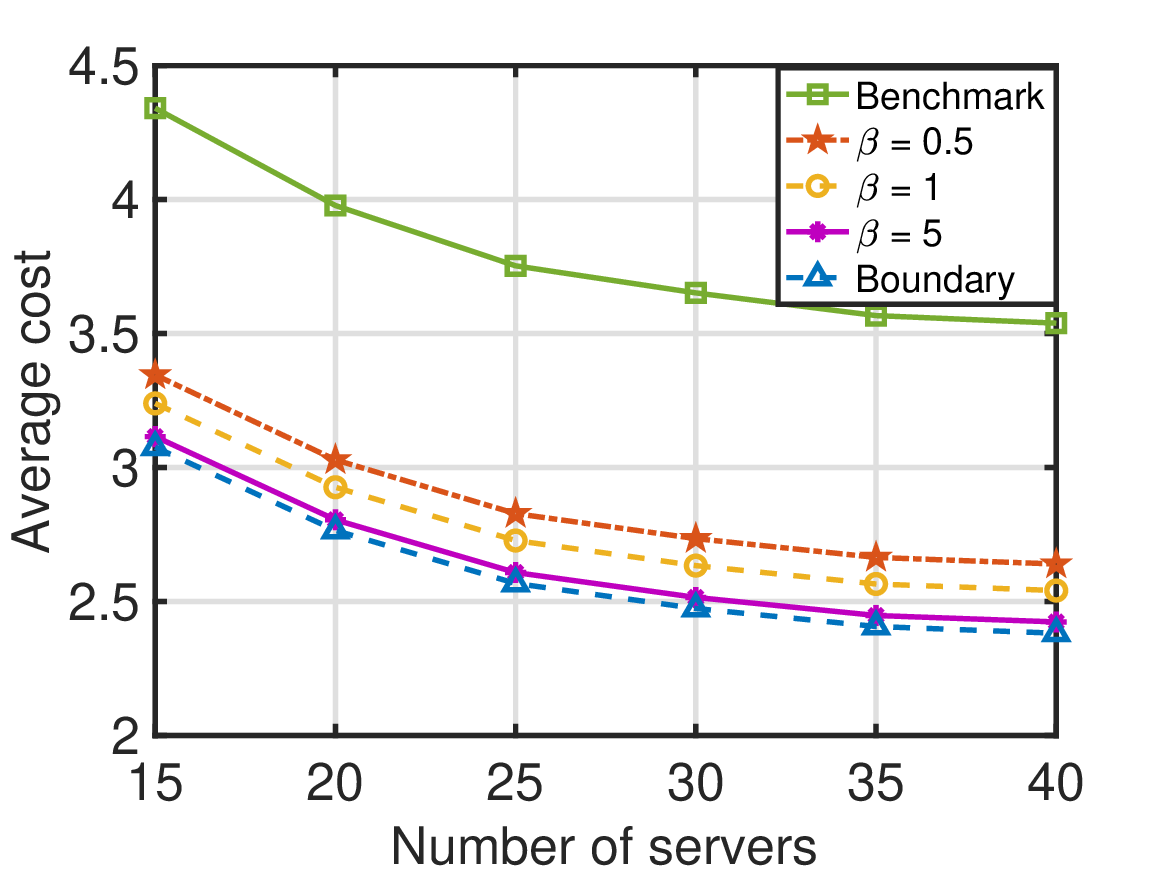}
        \caption{Average cost vs. The number of servers, $M$}
        \label{dt_cost}
    \end{minipage}
    \hfill
    \begin{minipage}[t]{0.24\textwidth}
        \includegraphics[width=\textwidth]{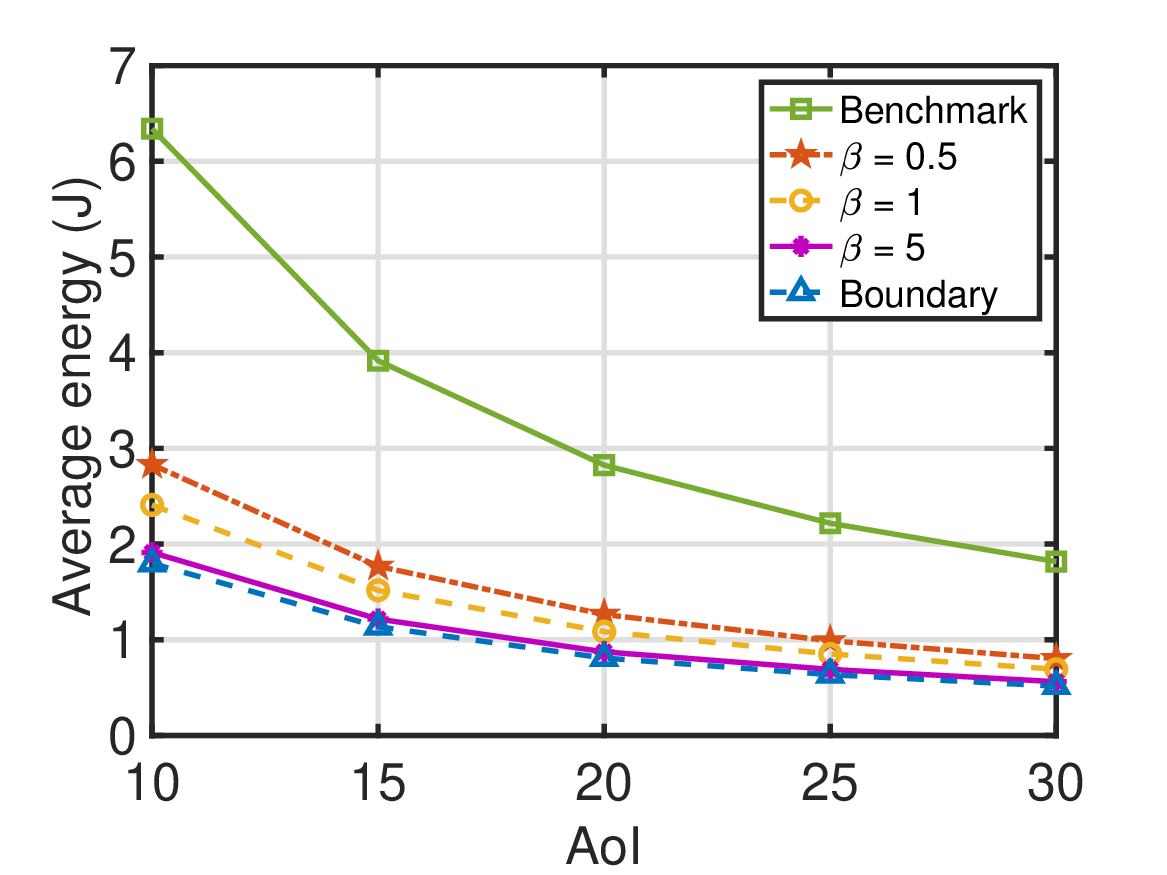}
        \caption{Average energy vs. AoI}
        \label{dt2}
    \end{minipage}
\hfill
    \begin{minipage}[t]{0.24\textwidth}
        \includegraphics[width=\textwidth]{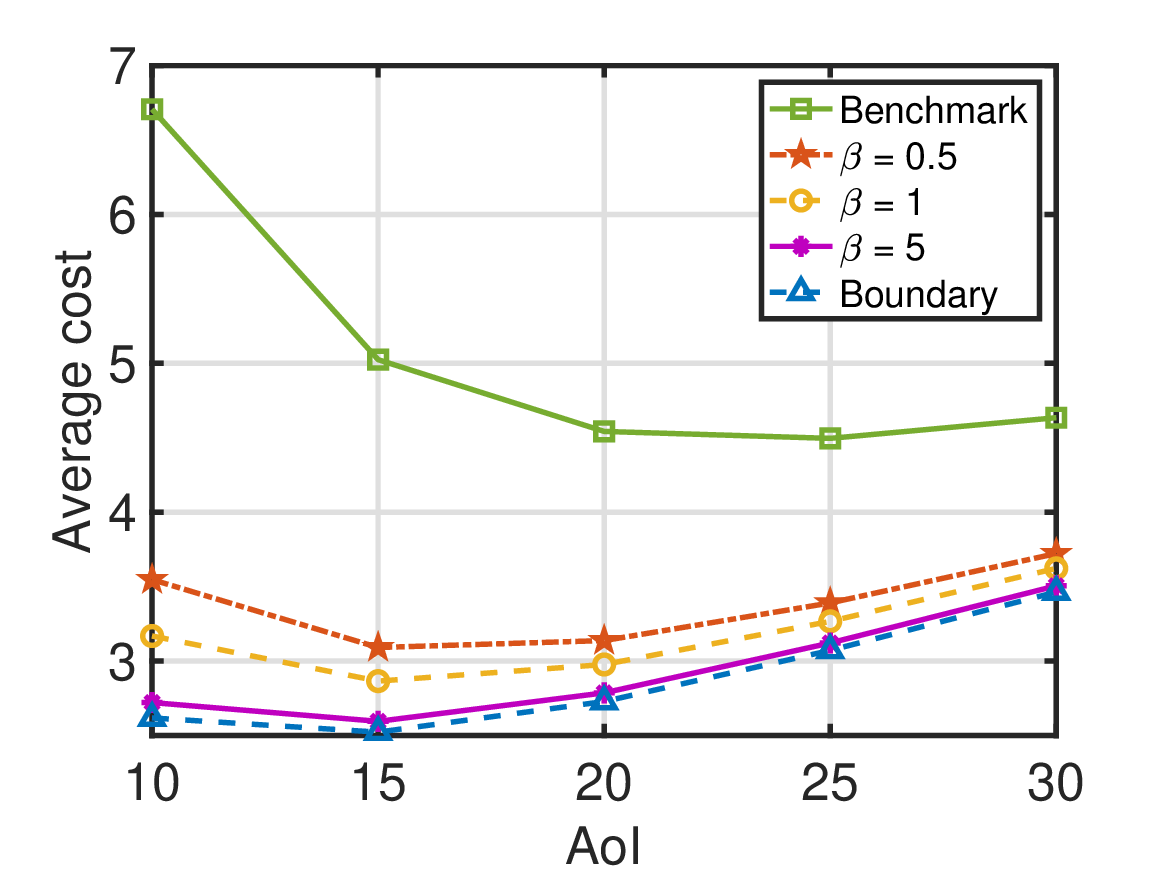}
        \caption{Average cost vs. AoI}
        \label{dt_cost2}
    \end{minipage} 
\end{figure*}

In this section, we evaluate the performance of our proposed joint optimization algorithm using Monte Carlo simulation. We consider a 1000m$\times$ 1000m square MEC network area, where both the edge servers and the devices are uniformly distributed. The synchronization data size of each device~$D_k$ is uniformly generated from $[2,5]$ MB. Similarly, the digital twin size for each device~$\tilde{D}_k$ is uniformly generated from the range $[5,50]$ MB~\cite{chen2024distributed}. We model the path loss effect by $128.1+37.6\log_{10}d$, where $d$ represents the distance between the servers and devices in kilometers. Furthermore, the wireless channel undergoes Rayleigh fading with a variance equal to 1. The noise power spectral density is $-174$ dBm/Hz. The system bandwidth is 10 MHz, and the duration of each time slot is $T_s=0.05$ s, with a total of 100 time slots considered. The backhaul transmission and migration costs are set to be $\eta= \lambda= 1\times 10^{-8}$ J/bit. During each time slot, the devices move randomly in a direction at speeds uniformly generated from $[2, 8]$ m/s. The parameter $\xi$ is established at 0.1. Each data point is averaged over 1000 realizations to ensure robustness and reliability in the evaluation process.

We evaluate the performance of our proposed scheme under different $\beta$ configurations, specifically setting $\beta$ to 0, 0.5, 1, and 5, respectively. When $\beta = 0$, it represents a benchmark scenario where digital twin migration occurs in every time slot. We refer to this scenario as ``Benchmark." Furthermore, we assign a significantly larger value to $\beta$ to delineate the boundaries of average energy consumption and cost through simulation, which we refer to as ``Boundary."
The number of devices~$K$ is 200. The maximum AoI constraint is set to be $\Gamma = 20$. The number of servers~$M$ varies from 10 to 50. We adopt the cyclic scheduling policy to make the average AoI equal to the maximum AoI constraint. Fig.~\ref{dt1} and Fig.~\ref{dt_cost} show the average energy and cost versus the number of servers, respectively. It can be seen that both the average energy and the average cost decrease as the number of servers increases. The reason is that a more significant number of servers provides more freedom of access. Besides, our proposed online scheduling algorithm saves more energy and costs with larger $\beta$. This is attributed to the fact that frequent migration leads to a large amount of energy consumption in this setting. Particularly, when there are 40 servers, our proposed scheme with $\beta=5$ showcases comparable performance to the Boundary approach and achieves energy savings of 21.7\%, 33.8\%, and 72.5\% compared to schemes with $\beta = 1$, $\beta=0.5$, and $\beta=0$, respectively.

Furthermore, we analyze the tradeoff between the AoI and energy. Considering 300 users and 30 servers, we vary the maximum AoI to be 10, 15, 20, 25, and 30. Fig.~\ref{dt2} and Fig.~\ref{dt_cost2} show the average energy and the average cost versus the maximum AoI, respectively. Fig.~\ref{dt2} shows the tradeoff between energy and AoI,  aligning with our earlier observations. The reason is that a smaller AoI necessitates more frequent user scheduling for digital twin synchronization, thereby amplifying energy consumption. Similarly, Fig.~\ref{dt_cost2} indicates that the average cost initially decreases and then increases as the AoI increases. When the maximum AoI is small, heightened energy consumption leads to energy costs dominating the average cost. Conversely, as the AoI increases, energy costs decrease, and the AoI becomes the dominant factor in the average cost.

\section{Conclusion}
In this paper, we studied the digital twin deployment and migration problem in MEC networks with the objective of minimizing the long-term average AoI and energy. To achieve this, the problems of edge association, power control, and digital twin deployment were jointly optimized. An optimal solution was proposed for the special case when $K=M\Gamma$ under a static channel. Furthermore, an efficient online lightweight algorithm was developed for the general case. Through comprehensive numerical analyses, we illustrated that our proposed approach, particularly with a higher parameter value $\beta$, can yield superior performance outcomes. How to utilize $\beta$ to regulate digital twin migration across various scenarios will be further investigated in our future work.

\bibliographystyle{IEEEtran}
\bibliography{IEEEabrv,References}
\end{document}